\documentclass[aps,prx,superscriptaddress,nofootinbib,twocolumn]{revtex4-2}

\usepackage[utf8]{inputenc}
\usepackage{amsmath,amssymb,amsthm, mathtools, physics}
\usepackage{enumitem}
\usepackage{dsfont}
\usepackage{braket}
\usepackage{bm}
\usepackage{graphicx, color, xcolor}
\usepackage{comment}
\usepackage[caption=false]{subfig}
\usepackage{hyperref}
\usepackage{float}
\hypersetup{
    colorlinks=true,
    citecolor=blue,
    linkcolor=magenta,
    urlcolor=magenta,
}

\newtheorem{lemma}{Lemma}
\newtheorem{theorem}{Theorem}
\newtheorem{definition}{Definition}

\makeatletter
\newsavebox{\@brx}
\newcommand{\llangle}[1][]{\savebox{\@brx}{$\m@th{#1\langle}$}%
  \mathopen{\copy\@brx\kern-0.5\wd\@brx\usebox{\@brx}}}
\newcommand{\rrangle}[1][]{\savebox{\@brx}{$\m@th{#1\rangle}$}%
  \mathclose{\copy\@brx\kern-0.5\wd\@brx\usebox{\@brx}}}
\makeatother

\begin{document}

\title{Efficient equivalence checking of Clifford-U circuits with shared single-qubit unitaries}

\author{Daisuke Sakamoto}
\email{u546555h@ecs.osaka-u.ac.jp}
\affiliation{School of Engineering Science, The University of Osaka, 1-3 Machikaneyama, Toyonaka, Osaka 560-8531, Japan}

\author{Soshun Naito}
\email{naito@biom.t.u-tokyo.ac.jp}
\affiliation{Department of Information and Communication Engineering, Graduate School of Information Science and Technology, The University of Tokyo, 7-3-1 Hongo, Bunkyo-ku, Tokyo 113-0033, Japan}

\author{Yusei Mori}
\affiliation{Graduate School of Engineering Science, The University of Osaka, 1-3 Machikaneyama, Toyonaka, Osaka 560-8531, Japan}

\author{Kosuke Mitarai}
\email{mitarai.kosuke.es@osaka-u.ac.jp}
\affiliation{Graduate School of Engineering Science, The University of Osaka, 1-3 Machikaneyama, Toyonaka, Osaka 560-8531, Japan}
\affiliation{Center for Quantum Information and Quantum Biology, The University of Osaka, 1-2 Machikaneyama, Toyonaka 560-0043, Japan}

\date{\today}

\begin{abstract}
Quantum circuit equivalence checking asks whether two circuits implement the same unitary. 
It guarantees compiler correctness and safe optimization, yet most existing approaches scale exponentially with the number of qubits or the circuit depth, or are restricted to specific circuit structures. 
In this work, we present an equivalence-checking method for circuits formed by arbitrary single-qubit layers interleaved with Clifford layers. 
This pattern is common in variational quantum algorithms and Hamiltonian simulation via Trotter decomposition.
It can also represent any unitary with sufficient depth. 
We prove the existence of an efficient classical algorithm that determines whether a pair of circuits with shared single-qubit layers are equivalent for every possible choice of the shared single-qubit unitaries.
The same algorithm can also certify their non-equivalence for fixed assignments of single-qubit unitaries.
Our framework supports the validation of emerging quantum compilers and facilitate the discovery of novel circuit optimization passes.
\end{abstract}

\maketitle

\section{Introduction}

\begin{table}[t]
    \centering
    \caption{Comparison of major circuit equivalence–checking methods and their worst-case computational costs to verify Eq. \eqref{eq:intro} with respect to the number of qubits $n$, gates $g$, and circuit depth $d$.}
    \begin{tabular}{c c c c}
        \hline\hline
        Method & complexity(worst-case)  \\ \hline
        Direct matrix computation & $4^n\cdot 2^{O(g)}$ \\
        PBEC~\cite{Yu2025Scalable} & $2^{\text{poly}(d)} \cdot O(n)$  \\
        ZX-calculus based~\cite{Peham2023EquivalenceCheckingParameterized} & $2^{O(n)} \cdot \text{poly}(g)$  \\
        This work & $O(n \cdot g)$  \\
        \hline\hline
    \end{tabular}
    \label{tab:complex}
\end{table}

Quantum computers are widely expected to provide computational advantages over their classical computers for certain classes of problems. 
However, extracting meaningful results from long-term quantum algorithms such as Shor's algorithm \cite{Shor1994} for integer factorization or Grover's search \cite{Grover1996} requires large-scale, fault-tolerant devices with deep circuits and robust error-correction schemes, far beyond the capabilities of hardware that is available today.

In contrast, quantum processors typically support only a limited gate library and exhibit limited qubit connectivity, meaning that algorithms containing high-level gates cannot be executed directly on actual devices. As a result, any practical implementation of a quantum algorithm necessarily relies on a compilation pipeline that synthesizes the circuit into the hardware-supported gate library, maps logical qubits to physical qubits, and applies optimizations to reduce the overall circuit size.

A key component for developing such optimization methods is equivalence checking, which is a task to determine whether two quantum circuits implement the same unitary operation. 
In full generality, exact equivalence checking is computationally intractable: it is NQP-hard \cite{NQP-hard}, and even approximate variants are QMA-hard \cite{QMA-hard}, including settings restricted to constant-depth circuits \cite{QMA-hard-ShortDepth}.
These hardness results motivate the development of practical frameworks that exploit additional circuit structure or employ efficient representations.

A useful way to view the landscape of practical equivalence checking is through the lens of classical simulation and classical representations of quantum circuits.
Many equivalence-checking frameworks are built around a simulator-like backbone that provides a compact representation and efficient manipulation rules for a target circuit family.
One example is equivalence checking frameworks for Clifford circuits \cite{Thanos2023FastEquivalenceClifford, Osama2025ParallelECGPU}, which leverages their classical simulability \cite{Aaronson2004}.
Two major approaches for this task is to utilize the ZX-calculus \cite{ZX-calculus} or the decision-diagram \cite{Miller2006QMDD} based representations of quantum circuits \cite{Peham2022EquivalenceParadigms,Peham2022EquivalenceCheckingZX, Peham2023EquivalenceCheckingParameterized, Peham2022EquivalenceParadigms, wang2008XQDDBasedVerification, Yamashita2008DDMFEfficientDecision, Wei2022AccurateBDDbasedUnitary, hong2022tdd, Wille2009EquivalenceReversible, Miller2006QMDD, Hong2024PQCEquivalence, Hong2022EquivalenceDQC, Choudhury2025ZXNet}.
Equivalence checking can also be supported by path-integral style simulators \cite{Amy2018LargeScaleVerification} or tensor-network based simulators \cite{Sander2025MPOEquivalence}.
Constraint-solving based approaches encode circuit semantics and equivalence as SAT/SMT constraints and use state-of-the-art SAT/SMT solvers \cite{Mei2024EquivalenceModelCounting, Berent2022SATQuantum, Chen2025AutomataQuantumVerification, Chen2025AutoQ2, BauerMarquart2023symQV, Xu2022Quartz, Shi2019CertiQ,Lewis2023FormalVerificationQuantum, Chareton2022FormalMethodsQuantumSurvey}.
Across these approaches, the core trade-off is similar: a representation is powerful on instances with specific structures, but the worst-case cost typically remains exponential for general circuits with respect to the number of qubits or gates.

In this work, we propose a method of equivalence-checking for a widely used and expressive class of quantum circuits in which layers of $n$-qubit Clifford gates $C$ alternate with arbitrary single-qubit unitaries $U$, i.e., 
\[
F(\vec{U}) = C_m U_m^{(k_m)} C_{m-1} U_{m-1}^{(k_{m-1})} \cdots C_1 U_1^{(k_1)} C_0.
\]
$F'(\vec{U})$, where each $U_i$ acts on a $k_i$-th single qubit and the Clifford layers ${C_i}$ may differ between the two circuits.
These ``Clifford-U'' circuits are known to possess universal expressiveness for quantum computation, and they are used as \emph{ansatz} for variational quantum algorithms (VQAs) \cite{cerezo2021vqa} and other parametric circuit architectures. 
Furthermore, this architecture naturally arises when implementing time evolution via the Trotter-Suzuki decomposition \cite{lloyd1996universal}, where the evolution time and Hamiltonian coefficients are mapped onto the continuous rotation angles of the single-qubit layers.

Instead of asking whether two fixed circuits $F$ and $F'$ are equivalent, we reformulate the problem as determining whether
\begin{align}
F(\vec{U}) = F'(\vec{U}) \quad \forall \, \vec{U},
\label{eq:intro}
\end{align}
i.e., whether the two \emph{ansatz} implement the same transformation for every possible assignment of the single-qubit unitaries $U_i$.
This reframing simplifies the computational complexity: by reducing the problem to the propagation of substituted Pauli operators through the Clifford layers, we show that the verification can be performed in polynomial time in both the number of qubits and the number of layers.
Table \ref{tab:complex} compares the complexity of the proposed method with that of existing algorithms.

The key idea is to substitute Pauli operators in place of the arbitrary unitaries $U_i$
and track their adjoint action through the circuit, which is efficient on classical computers by Gottesman-Knill theorem \cite{Aaronson2004}.
This substitution allows us to derive algebraic conditions that are necessary and sufficient for the two circuits to be equivalent for all possible choices of the unitaries $\vec{U}$; consequently, these conditions yield an efficient and complete procedure for equivalence checking within this class of circuits.

Moreover, although our method is framed in terms of equivalence for all possible $\vec{U}$, it also provides a practical and computationally lightweight approximate indicator of equivalence for the original fixed circuits $F(\vec{U})$ and $F'(\vec{U})$.
Thus, the approach simultaneously offers a rigorous polynomial-time solution for parametric equivalence checking and a fast heuristic for ordinary circuit equivalence.

The remainder of this paper is organized as follows. We first formalize the
structural assumptions under which our method applies and state the precise
conditions for equivalence. We then present the theoretical proof of
correctness and describe the resulting polynomial-time verification
procedure in detail. Next, we compare our approach with existing methods and
discuss its conceptual and practical advantages. Finally, we conclude by
summarizing our contributions and outlining potential extensions.

\section{Problem Setting and Main Result}
\label{sec:Problem-Result}

\begin{figure*}
  \centering
  \includegraphics[height=7em]{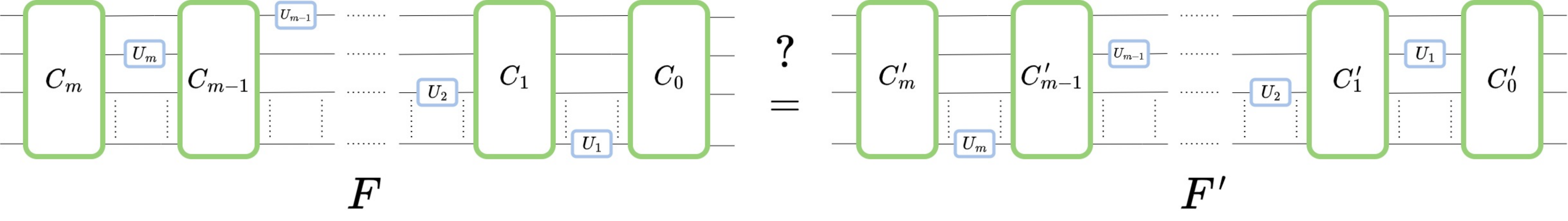}
  \caption{Circuit diagrams of the quantum circuits $F$ and $F'$ targeted for equivalence checking in this paper.}
  \label{fig:F,Fprime}
\end{figure*}

In this section, we formally specify the class of quantum circuits considered in this work and present our main equivalence-checking result.
We first introduce the definitions required to state our main result, and then present the central theorem, which gives a complete and efficient
criterion for equivalence-checking.

Throughout this paper, we use $\mathcal{U}_d$ for a set of $d \times d$ unitary matrix, $\mathcal{P}=\{I, X, Y, Z\}$ for a set of Pauli operator.
Also, for a single-qubit unitary $U\in \mathcal{U}_2$, the notation $U^{(k)}$ represents  $U$ acting on the $k$-th qubit of the $n$-qubit system. 
Explicitly, it is defined as the tensor product
\begin{equation}
    U^{(k)} := I^{\otimes (k-1)} \otimes U \otimes I^{\otimes (n-k)},
\end{equation}
where $I$ denotes the $2 \times 2$ identity matrix.

\begin{definition}
\label{def:F,Fprime}
Let $F,F' : \mathcal{U}^m_d \to \mathcal{U}_{2^n}$ be functions that take as input a sequence of single-qubit unitaries
$\vec{U} = [U_1, \ldots, U_m]$ and output the corresponding $n$-qubit quantum circuit (unitary operator) defined as
\begin{align}
  F(\vec{U}) &= C_m U_m^{(k_m)} C_{m-1} U_{m-1}^{(k_{m-1})} \cdots C_1 U_1^{(k_1)} C_0,
  \label{eq:def-circuit_1}  \\
  F'(\vec{U}) &= C'_m U_m^{(l_m)} C'_{m-1} U_{m-1}^{(l_{m-1})} \cdots C'_1 U_{1}^{(l_1)} C'_0,
  \label{eq:def-circuit_2}
\end{align}
where $C_i$ and $C'_i$ are $n$-qubit Clifford gates.
\end{definition}
Note that, in the above definition, the pair of sequences $(\{C_i\}_i, \{k_i\}_i)$ and $(\{C'_i\}_i, \{l_i\}_i)$ are fixed parameters characterizing $F$ and $F'$, respectively.
Figure~\ref{fig:F,Fprime} shows the circuit diagrams for the quantum circuits $F$ and $F'$ defined in Definition~\ref{def:F,Fprime}.

Our goal is to determine whether the two circuit \emph{functions} $F$ and $F'$ implement the same unitary transformation for \emph{all} possible assignments of the single-qubit unitaries $\vec{U}$, up to a global phase.
To analyze this problem, we study how Pauli operators propagate through the Clifford layers of each circuit.
\begin{definition}
\label{def:V}
For $C_i$ and $C_i'$ in Definition~\ref{def:F,Fprime}, 
let $C_{a:b} := C_a C_{a-1} \cdots C_b$ where $a \ge b$ and 
$C'_{a:b} := C'_a C'_{a-1} \cdots C'_b$. 
We define 
\begin{align}
V_i(P_i) := 
C_{m:i}\, P_i^{(k_i)}\, C_{m:i}^\dagger, 
\\
V'_i(P_i) := 
C'_{m:i} P_i^{(l_i)} C_{m:i}'^\dagger.
\label{eq:def-V}
\end{align}
Note that $V_i(P_i)$ can be written as $V_i(P_i) = (-1)^{r_i} P_1\otimes P_2\otimes \cdots \otimes P_n$ for some $r_i\in \{0,1\}$ and $P_i \in \mathcal{P}$.
Also, we denote the binary symplectic representation of $V_i(P_i)$  by the vector
\begin{align}
\bm{v}_i(P_i) :=
(z_{1i}, \ldots, z_{ni}, x_{1i}, \ldots, x_{ni} , r_i)^\top
\in \mathbb{F}_2^{2n+1}.
\label{eq:v-vector-def}
\end{align}
The pair $(z_{qi}, x_{qi}) \in \mathbb{F}_2^2$ is determined by $P^{(q)}$ according to the following correspondence:
\begin{align*}
P^{(q)} = I \Rightarrow (z_{qi}, x_{qi}) = (0,0), \\
P^{(q)} = X \Rightarrow (z_{qi}, x_{qi}) = (0,1), \\
P^{(q)} = Y \Rightarrow (z_{qi}, x_{qi}) = (1,1), \\
P^{(q)} = Z \Rightarrow (z_{qi}, x_{qi}) = (1,0).    
\end{align*}
\end{definition}

These operators describe how a Pauli substitution at a given layer propagates through the remainder of the Clifford circuit.
Note that each of $\bm{v}_i(P_i)$ can be computed on a classical computer in $\mathrm{poly}(nm)$ time using, for example, the algorithm proposed in Ref. \cite{Aaronson2004}.
Using these vectors, we collect the propagation data of all Pauli substitutions into a single matrix.
\begin{definition} [Decision matrix]
\label{def:B,Bprime}
For the circuit $F(\vec{U})$ defined as Definition \ref{def:F,Fprime}, 
let $S_X$ and $S_Z$ respectively be the horizontal concatenation of the $m$ vectors $\bm{v}_i(X)$ and $\bm{v}_i(Z)$ corresponding to
$V_i(X)$ for $V_i(Z)$ defined in Definition \ref{def:V}.
We define a $(2n{+}1) \times 2m$ matrix $B$ for the circuit $F$ as
\begin{align}
B :&=
\begin{bmatrix}
S_Z \mid S_X 
\end{bmatrix} \notag \\ 
&=
\begin{bmatrix}
z_{11} & \cdots & z_{1m} & z_{1(m+1)} & \cdots & z_{1(2m)} \\
\vdots & \ddots & \vdots & \vdots & \ddots & \vdots \\
z_{n1} & \cdots & z_{nm} & z_{n(m+1)} & \cdots & z_{n(2m)} \\
x_{11} & \cdots & x_{1m} & x_{1(m+1)} & \cdots & x_{1(2m)} \\
\vdots & \ddots & \vdots & \vdots & \ddots & \vdots \\
x_{n1} & \cdots & x_{nm} & x_{n(m+1)} & \cdots & x_{n(2m)} \\
r_{1} & \cdots & r_{m} & r_{m+1} & \cdots & r_{2m} 
\end{bmatrix}.
\label{eq:decision-matrix}
\end{align}
We refer to this matrix $B$ as the \emph{decision matrix}.
A decision matrix $B'$ is analogously defined for $V'_i(P_i)$.
\end{definition}
We are now ready to state our main result, which shows that equivalence for all parameter assignments can be decided by comparing two efficiently computable objects.
\begin{theorem}
\label{thm:equiv-theory}
Under definition~\ref{def:F,Fprime},~\ref{def:V}, and~\ref{def:B,Bprime}, $F(\vec{U})$ and $F'(\vec{U})$ are equivalent for any $\vec{U}$, \emph{if and only if} these two equations
\begin{align}
B' &= B,
\label{eq:equation-BB} \\
  C'_{m:0} &= e^{i\phi} C_{m:0},
\label{eq:CC}
\end{align}
simultaneously hold.
\end{theorem}
The theorem reduces the equivalence-checking problem for a broad class of quantum circuits to a comparison of their Clifford backbones and their associated decision matrices.

We stress that the computation of $\bm{v}_i(P_i)$ and $B$, and thus the entire equivalence-checking process based on Theorem~\ref{thm:equiv-theory}, can be performed efficiently on a classical computer. 
Specifically, by utilizing the standard stabilizer tableau formalism \cite{Aaronson2004}, the computational complexity of the verification process is strictly bounded by $O(n \cdot g)$.
Here, $g$ denotes the total number of Clifford gates contained in the quantum circuit $F$, and $n$ is the number of qubits.
This guarantees that our verification framework scales polynomially, avoiding the exponential bottleneck typical of general quantum circuit simulation.

It is notable that Eq.~\eqref{eq:equation-BB} can also be used as an efficient test for equivalence of circuits instantiated with a \emph{fixed} choice of parameters, i.e.,
\begin{align}
    F(\vec{U}) = e^{i\phi}F'(\vec{U}), \qquad \exists \, \vec{U} \in \mathcal{U}(2).
\end{align}
In this setting, we deliberately allow for false negatives—cases in which equivalent circuits are reported as inequivalent—in exchange for a significant reduction in computational cost.
Crucially, however, the method admits \emph{no false positives}:
whenever the procedure reports two circuits as equivalent, they are guaranteed to implement the same unitary.
Thus, the present theory provides both a rigorous polynomial-time criterion for parametric equivalence and a sound, compiler-friendly heuristic for ordinary circuit equivalence.

\section{Proof Sketch}
\label{sec:Structure}

In this section, we provide a proof sketch of the main theorem and explain the logical structure underlying our equivalence criterion.
Complete proofs of all intermediate statements are deferred to the Appendix~\ref{app:non-rearrange}.

The first key observation is that the equivalence of parameterized circuits for all single-qubit unitaries can be reduced to their equivalence under Pauli substitutions.
By expanding each single-qubit unitary in the Pauli basis, we show that 
\begin{align}
    F'(\vec{U}) = e^{i\phi}F(\vec{U}), \qquad \forall \, \vec{U} \in \mathcal{U}^m,
\label{log:equiv-U}
\end{align}
holds for a global phase $\phi$ if and only if 
\begin{align}
    F'(\vec{P}) = e^{i\phi}F(\vec{P}), \qquad \forall \, \vec{P} \in \mathcal{P}^m.
\label{log:equiv-P}
\end{align}
Furthermore, for a specific instance of single-qubit unitaries $\vec{U}$, Eq.~\eqref{log:equiv-P} provides a \emph{sufficient} criterion for the circuit equivalence $F'(\vec{U}) = e^{i\phi}F(\vec{U})$.

To analyze Eq.~\eqref{log:equiv-P}, we leverage the fact that Clifford circuits conjugate Pauli operators to Pauli operators.
We define the adjoint stabilizer operator 
\begin{align}
    V_i(P_i) = C_{m:i} P_i^{(k_i)} C_{m:i}^\dagger ,
\end{align}
where $C_{m:i} = C_m \cdots C_i$ represents the cumulative Clifford operation from layer $i$ to $m$.
By propagating the Pauli operators to the end of the circuit, the Pauli-substituted circuit $F(\vec{P})$ can be rewritten as:
\begin{align}
F(\vec{P}) = V_m(P_m) \cdots V_1(P_1) C_{m:0}. 
\label{log:def-FVC}
\end{align}
Similarly, the second circuit $F'(\vec{P})$ can be expressed using its corresponding operators $V'_i(P_i)$ and $C'_{m:0}$ as:
\begin{align}
F'(\vec{P}) = V'_m(P_m) \cdots V'_1(P_1) C'_{m:0}.
\label{log:def-FVC-prime}
\end{align}
From Eqs.~\eqref{log:def-FVC}, \eqref{log:def-FVC-prime}, it is evident that if the individual operators match layer-by-layer,
\begin{align}
V_i(P_i) = V'_i(P_i), \quad \forall \,i,
\label{log:equiv-V}
\end{align} 
and the cumulative Clifford operations agree,
\begin{align}
C'_{m:0} = e^{i\phi}C_{m:0},
\label{log:equiv-CC}
\end{align}
then Eq.~\eqref{log:equiv-P} naturally holds.
Thus, it is shown that the converse is also true.

Finally, since each adjoint stabilizer operator $V_i(P_i)$ is essentially a Pauli string up to a sign, it admits a binary symplectic representation. 
Collecting these vector representations into the decision matrix $B$, we show that the layer-wise equality (Eq.~\eqref{log:equiv-V}) is equivalent to the equality of the decision matrices $B = B'$. 
This reduces the exponential verification problem to an efficient polynomial-time matrix comparison.

The same logical structure also applies when the single-qubit unitaries appear in different orders in the two circuits, i.e.,
\begin{align}
F(\vec{U}) &:= C_m U_m^{(k_m)} C_{m-1} U_{m-1}^{(k_{m-1})} \cdots C_1 U_1^{(k_1)} C_0, \notag \\
F'(\vec{U_\sigma}) &:= C'_m U_{\sigma(m)}^{(l_m)} C'_{m-1} U_{\sigma(m-1)}^{(l_{m-1})} \cdots C'_1 U_{\sigma(1)}^{(l_1)} C'_0. \notag
\end{align}
The detailed treatment of such permutations is given in the Appendix~\ref{app:rearrange}.

\section{Discussion}

\subsection{Applicability}
Our method applies to quantum circuits in which all non-Clifford operations are single-qubit unitaries that appear in both circuits.
As a result, transformations that rely on decomposing or reconstructing multi-qubit non-Clifford gates are outside the scope of the present framework.

Despite this restriction, the considered circuit class is expressive enough to represent arbitrary unitaries.
Circuits composed of alternating Clifford layers and single-qubit unitaries are known to be universal and therefore serve as standard \emph{ansatz} in variational quantum algorithms (VQAs) \cite{cerezo2021vqa} and related parametric models.
The proposed method is particularly well suited for equivalence-checking in the VQA setting.
In VQAs, the \emph{ansatz} structure is fixed before compilation, and the single-qubit gates are treated as symbolic parameters;
correctness therefore requires equivalence of the compilation to hold for \emph{all} parameter values.
Our framework provides an efficient and complete verification procedure precisely for this parametric equivalence.

Another important application of our framework is the verification of Hamiltonian-simulation circuits based on Trotter decomposition \cite{lloyd1996universal}. 
In quantum simulations of many-body systems, the time-evolution operator is commonly compiled into a sequence of elementary gates with continuous rotation parameters determined by the evolution time and the Hamiltonian coefficients. 
These circuits naturally exhibit the Clifford-U structure, making them directly amenable to our method. 
This is particularly useful in practice, since Trotterized circuits often undergo extensive optimization to reduce gate overhead and circuit depth \cite{wecker2014gate, childs2018toward, nam2018automated}. 
Our framework then provides a rigorous means of certifying that such optimized implementations remain functionally equivalent to the original Trotter sequence for all parameter values, rather than only for selected parameter instantiations. 
As a result, it offers a guarantee that optimaization preserves the intended dynamics exactly.

\subsection{Limitations Imposed by NQP- and QMA-Hardness Results}
While our method provides strong guarantees for parametric equivalence, it does not fully solve equivalence checking for a specific quantum circuit.
This limitation reflects the intrinsic computational hardness (NQP-hard, QMA-hard) of equivalence-checking of quantum circuits once non-Clifford gates are involved.
Even so, when used to verify circuits under a specific parameter assignment $\vec{U}$, our method remains highly efficient. 
It operates as a filter that will never mistakenly validate non-equivalent circuits, though it might fail to recognize two equivalent circuits if their equivalence relies entirely on those specific parameter values.

However, the probability of such an accidental coincidence is rigorously bounded by the mathematical properties of parameterized circuits. 
As demonstrated by Peham et al., for any two circuits that are not functionally equivalent, the probability that a random parameter instantiation yields identical unitary operations is strictly zero \cite{Peham2023EquivalenceCheckingParameterized}. 
This stems from the fact that the matrix elements of these circuits are analytic functions of the rotation angles, and the set of zeros for non-identical analytic functions constitutes a set of measure zero. 
Therefore, our method's reliance on functional equivalence is well-justified: any structural discrepancy identified by our decision matrices represents a non-equivalence that will persist for almost all possible parameter values.

\subsection{Advantages over Heuristic and Local Verification Methods}

Our method detects global circuit equivalences that cannot be captured by a naive layer-by-layer comparison of individual Clifford operations.
Even when corresponding Clifford layers differ, nontrivial cancelations across layers may yield identical overall transformations, which local checks fail to identify.
The proposed approach is also robust to reordering of single-qubit unitaries and changes in their acting qubits.
Unlike heuristic rewrite-based techniques, our method follows a fixed, deterministic procedure.
This structural robustness suggests that the framework remains applicable even as compiler architectures evolve.
Furthermore, this deterministic approach may serve as a foundation for discovering new, provably rewrite principles.

\section{numerical experiment}
To empirically assess that the equivalence-checking procedure implied by Theorem~\ref{thm:equiv-theory} can be executed in polynomial time in the number of qubits and the number of layers, we conducted numerical experiments.
We generated quantum circuits using PyZX(0.8.0)~\cite{pyzx} and applied the built-in optimization routines of the same library. 
We then tested the equivalence between each optimized circuit and its original circuit, following Theorem~\ref{thm:equiv-theory}, and measured the runtime required for this decision.
Our proposed method was implemented using Qiskit(1.4.5)~\cite{qiskit}. 
Circuit instances were stored in QASM(1.0.1)~\cite{qasm}.
In addition, to validate the correctness of the outcomes produced by our method and to benchmark runtimes, we also performed equivalence checking using MQT-qcec (QCEC~2.8.2), a tool provided by the Munich Quantum Toolkit~\cite{2021qcec, 2021QCECadvanced}. 
QCEC enables fast equivalence checking by combining multiple techniques, including decision diagrams and the ZX-calculus.

All experiments were performed on a MacBook Air.
The machine was equipped with an Apple M4 CPU and $24$~GB of memory, and ran macOS~15.5.

\subsection{Setup}
In our experiments, we used as a benchmark a parametric quantum circuit $F$ of the form in Eq.~\eqref{eq:experiment_situation}:
\begin{align}
  F(\vec{U}) &= C_m U_m^{(k_m)} C_{m-1} U_{m-1}^{(k_{m-1})} \cdots C_1 U_1^{(k_1)} C_0 .
  \label{eq:experiment_situation}
\end{align}
Here, each $C_i$ is an $n$-qubit Clifford circuit generated using PyZX's \texttt{CNOT\_HAD\_PHASE\_circuit} routine.
Specifically, we constructed depth-$10$ Clifford circuits whose gate composition consists of CNOT, Hadamard, and phase gates in the ratio $20\%$, $40\%$, and $40\%$, respectively, and used them as the $C_i$.
Each $U_i^{(k_i)}$ is a single-qubit gate acting on the $k_i$-th qubit, constructed as a product of rotations around the $Z$, $X$, and $Z$ axes.
The rotation angles for these gates were chosen uniformly at random from the interval $[0, 2\pi)$.
The index $k_i$ was chosen as a randomly generated integer.

As a comparison instance, we constructed a circuit $F'$ that is designed to be equivalent to $F$. 
To generate $F'$, we retained the identical single-qubit rotation gates $U_i$ but assigned them to new qubit indices $l_i$, which were sampled independently at random and separately from the original indices $k_i$. 
To preserve the overall equivalence of the circuit under this qubit reassignment, each corresponding Clifford layer $C'_i$ was constructed by padding the original $C_i$ with appropriate SWAP gates. 
Subsequently, we applied an optimization step to these modified Clifford layers using PyZX (\texttt{optimize\_process\_pyzx}). 
This systematic construction yields the following circuit:
\begin{align}
  F'(\vec{U}) &= C'_m U_m^{(l_m)} C'_{m-1} U_{m-1}^{(l_{m-1})} \cdots C'_1 U_1^{(l_1)} C'_0,
\end{align}
which is globally equivalent to $F$ by design.

We also generated a non-equivalent instance $F_{\mathrm{err}}$ and performed equivalence checking against $F$.
To construct the circuit $F_{\mathrm{err}}$, we took the previously generated circuit $F'$ which is equivalent to $F$ and appended to it a single Pauli gate (either $X$ or $Z$) chosen uniformly at random.
The insertion location was chosen as the end of a randomly selected $C'_i$, and the target qubit was also selected uniformly at random.
This Pauli insertion is intended to artificially emulate a compilation-time error. 
Consequently, $F_{\mathrm{err}}$ is expected to be theoretically non-equivalent to $F$.

Finally, we constructed another non-equivalent comparison circuit $G$ and likewise tested it for equivalence with $F$:
\begin{align}
  G(\vec{U}) &= D_m U_m^{(q_m)} D_{m-1} U_{m-1}^{(q_{m-1})} \cdots D_1 U_1^{(q_1)} D_0 .
\end{align}
The circuit $G$ is composed of $n$-qubit Clifford circuits $D_i$ generated a new using \texttt{CNOT\_HAD\_PHASE\_circuit}.
The single-qubit gates $U_i$ use the same rotation angles as in $F$, while the target-qubit indices $q_i$ were sampled uniformly at random.
With this construction, except for the negligible probability of accidentally generating an equivalent circuit, $G$ is expected to be theoretically non-equivalent to $F$.

\subsection{Experimental procedure}

Because existing state-of-the-art tools such as QCEC exhibit exponential runtime scaling, performing systematic comparisons for large quantum circuits is computationally prohibitive. 
To effectively characterize both relative performance and scalability, we divided our numerical evaluation into two distinct regimes: a small-scale comparison ($5 \le n \le 10$) to directly benchmark our proposed method against QCEC, and a large-scale evaluation ($10 \le n$) to demonstrate the polynomial scaling of our approach.

In the small-scale regime, we measured the verification runtimes across various number of layers $m$ for equivalent pairs ($F, F'$), as well as for non-equivalent pairs constructed by either injecting local errors ($F, F_{\text{err}}$) or generating independent circuits ($F, G$).
For the large-scale regime, we restricted our measurements exclusively to the proposed method and focused entirely on equivalent circuit pairs ($F, F'$). 
This restriction is algorithmically motivated by the evaluation process of Theorem~\ref{thm:equiv-theory}. 
To verify the equivalence, our method compares the decision matrices (Definition~\ref{def:B,Bprime}) column by column, which corresponds to checking the matching condition for each operator $V$ (Definition~\ref{def:V}) as expressed in Eq.~\eqref{log:equiv-V}. 
The moment a single mismatch is detected between any corresponding columns, the algorithm immediately terminates and declares the circuits non-equivalent.
Consequently, verifying equivalent circuits necessitates the exhaustive evaluation of all columns, thereby manifesting the worst-case computational complexity of the algorithm.
By subjecting our method to these worst-case instances, we aim to definitively confirm its practical efficiency for circuit sizes well beyond the reach of conventional verification techniques.

\section{Results of numerical experiments}

\subsection{Runtime comparison between the proposed method and QCEC (Experiment~1)}

\begin{figure}[t]
  \centering
  \includegraphics[height=40em]{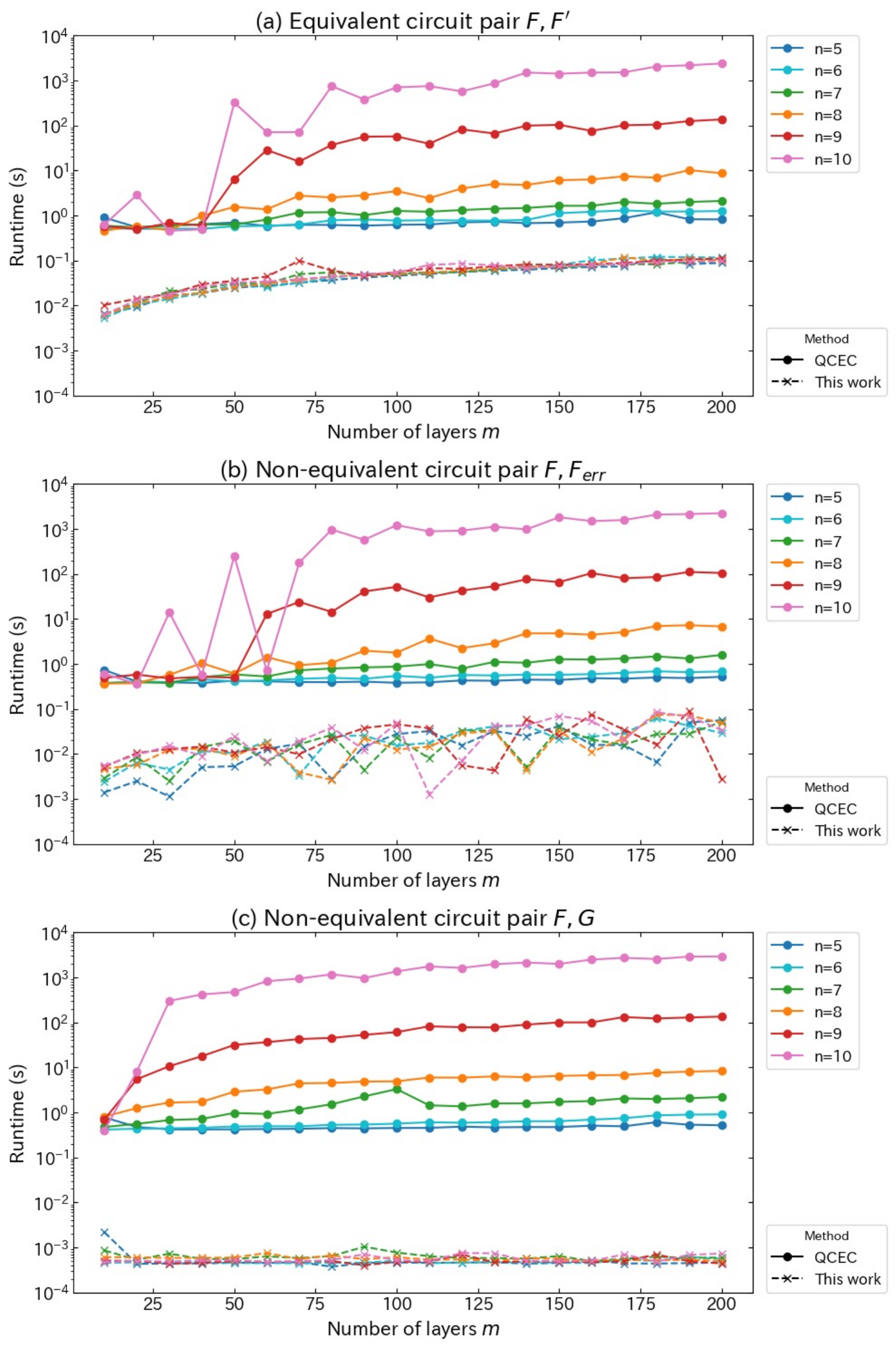}
  \caption{Relationship between the number of layers $m$ of the circuit pair and the runtime of the proposed method and QCEC (logarithmic scale on the vertical axis).
  We evaluate the runtime of these methods for pairs of equivalent circuits in (a) and non-equivalent circuits in (b) and (c). In (b), $F_\text{err}$ was generated by inserting an Pauli error into $F$. In (c), $F$ and $G$ are generated independently.}
  \label{fig:time-comparision}
\end{figure}

We first performed equivalence checking for the equivalent circuit pair $F$ and $F'$.
For all instances, both QCEC and the proposed method concluded that the circuits are equivalent.
The measured runtimes of the proposed method and QCEC are shown in Fig.~\ref{fig:time-comparision}(a).
These results confirm that, although the runtime of QCEC increases exponentially as the number of qubits grows, the runtime of the proposed method increases only marginally.

We next performed equivalence checking for the non-equivalent circuit pair $F$ and $F_{\text{err}}$.
For all instances, both QCEC and the proposed method concluded that the circuits are non-equivalent.
The measured runtimes of the proposed method and QCEC are shown in Fig.~\ref{fig:time-comparision}(b).
Comparing this result with Fig.~\ref{fig:time-comparision}(a), we see that the runtime required by QCEC is roughly the same regardless of whether the input circuits are equivalent.
In contrast, the runtime of the proposed method becomes more variable than in the equivalent-case experiments.
This variability arises because the column at which the decision matrices $B$ and $B'$ (defined in Definition~\ref{def:B,Bprime}) first disagree depends on the layer of $F'$ into which the Pauli gate is inserted.
The proposed method declares non-equivalence as soon as it detects a mismatching column while constructing $B$ and $B'$.
Therefore, the runtime of the proposed method varies depending on when the first mismatching column is encountered.

We further performed equivalence checking for the non-equivalent circuit pair $F$ and $G$.
For all instances, both QCEC and the proposed method concluded that the circuits are non-equivalent.
The measured runtimes of the proposed method and QCEC are shown in Fig.~\ref{fig:time-comparision}(c).
Comparing this result with Fig.~\ref{fig:time-comparision}(a), we find that the runtime required by QCEC does not differ substantially depending on whether the input circuits are equivalent.
On the other hand, the runtime of the proposed method is drastically shorter than in the equivalent-case experiments.
We attribute this to the fact that $F$ and $G$ are generated independently, so their Clifford gates differ in most layers.
As a consequence, the proposed method can often declare non-equivalence immediately after constructing and comparing just a single column of the decision matrices $B$ and $B'$.

\subsection{Observing the runtime of the proposed method (Experiment 2)}

\begin{figure}[t]
  \centering
  \includegraphics[height=27em]{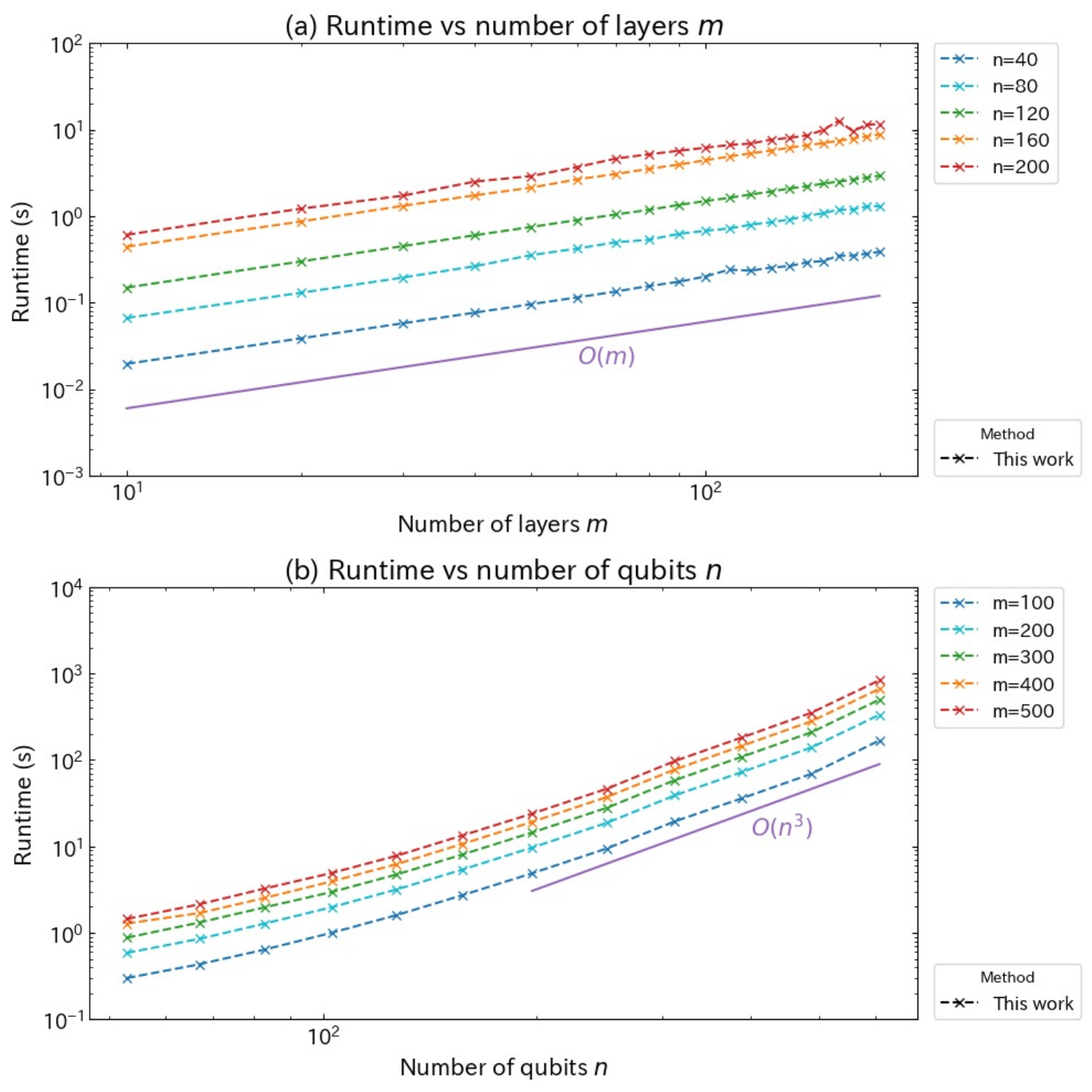}
  \caption{Relationship between the number of layers~$m$ of the equivalent circuit pair $F,F'$ and the runtime of the proposed method (log-log plot).
  We evaluate the runtime of the proposed method on circuit pairs. In (a), we measure the runtime as a function of the number of qubits $n$. In (b), we measure the runtime as a function of the number of layers $m$.}
  \label{fig:Scalable}
\end{figure}

For all instances of the equivalent circuit pair $F$ and $F'$, the proposed method successfully verified their equivalence.
First, Fig.~\ref{fig:Scalable}(a) shows the runtime dependence on the number of layers $m$ for $n \in \{40, 80, 120, 160, 200\}$.
The results confirm that the runtime scales strictly linearly ($O(m)$) with the number of layers, maintaining this efficiency up to 200 qubits.

Furthermore, Fig.~\ref{fig:Scalable}(b) illustrates the runtime dependence on the number of qubits $n$ on a log-log scale.
To eliminate hardware artifacts such as cache aliasing, we deliberately chose $n$ as logarithmically distributed prime numbers between $53$ and $607$.
While Python execution overheads dominate at small $n$, the curves become distinctly linear for $n \ge 197$, clearly demonstrating an empirical $O(n^3)$ scaling with respect to the number of qubits.

\subsection{Discussion of the experimental results}

Our experiments confirm that the proposed method verifies Clifford-$U$ circuits in polynomial time. 
Theoretically, the optimal time complexity is bounded by $O(g \cdot n)$, where $g$ is the total number of gates. 
Since the depth of each Clifford layer in our setup is constant ($d=10$), the total number of gates scales as $g \propto m \cdot n$. 
Consequently, the theoretically optimal runtime is expected to scale as $O(m \cdot n^2)$.

The theoretical $O(m)$ scaling was perfectly captured in Experiment 2. 
The observed $O(n^3)$ scaling with qubits, deviating from the optimal $O(n^2)$, stems entirely from our implementation. 
Specifically, Qiskit's \texttt{Clifford} class relies on dense binary matrix multiplications. 
Developing a dedicated implementation with optimized bit-wise tableau updates would resolve this overhead and achieve the optimal $O(n^2)$ scaling.

Nevertheless, our method's robust polynomial-time scaling provides a decisive practical advantage. 
As demonstrated in Experiment 1, while general-purpose tools suffer from exponential runtime explosions as the qubit count increases, our structural approach enables fast and scalable equivalence checking for complex quantum circuits.

\section{Conclusion}
In conclusion, we have proposed a novel equivalence-checking framework for Clifford-$U$ circuits. 
By expanding the single-qubit unitaries into the Pauli basis and exploiting the underlying algebraic structure of each term, our method leverages the Clifford tableau formalism to achieve verification in polynomial time. 
This approach provides a necessary and sufficient criterion for parametric equivalence, guaranteeing that two circuits are functionally identical for all possible single-qubit gate assignments. 
Furthermore, for fixed parameter instantiations, it serves as a safe verification.

By utilizing the proposed method, one can rigorously perform complete verification on the compilation results of parameterized quantum circuits used in variational quantum algorithms (VQAs) and circuits derived from Trotter-Suzuki decompositions. 
Moreover, we anticipate that this framework will contribute significantly to the validation of emerging quantum compilers and facilitate the discovery of novel circuit optimization passes.

\begin{acknowledgments}
This work is supported by MEXT Quantum Leap Flagship Program (MEXT Q-LEAP) Grant No. JPMXS0120319794, JST COI-NEXT Grant No. JPMJPF2014, and JST Moonshot R\&D Grant No. JPMJMS256J.
K.M. is supported by JST FOREST Grant No. JPMJFR232Z and JSPS KAKENHI Grant No. 23H03819.
Y.M. is supported by JSPS KAKENHI Grant No. JP25KJ1712.
\end{acknowledgments}

\appendix

\section{Proof of Theorem}
\label{app:non-rearrange}
In this Appendix, we prove the main equivalence criterion.
We first show that equivalence for all assignments of single-qubit unitaries is equivalent to equivalence under Pauli substitutions.
We then rewrite the Pauli-substituted circuits, and separates the $\vec{P}$-dependent and $\vec{P}$-independent parts of the equivalence condition.
Finally, we reduce equality of $W(\vec{P})$ for all $\vec{P}$ to a layerwise condition on the adjoint images $V_i$, and express that condition as equality of the Decision Matrices $B$ and $B'$.

\begin{definition}[Cumulative Adjoint Products]
\label{def:W}
Based on the adjoint images of the substituted single-qubit Pauli operators introduced in Definition~\ref{def:V},  we also define the cumulative adjoint products
\begin{align}
W(\vec{P}) = V_m(P_m) \dots V_1(P_1), 
\label{eq:def-W-1}
\\
W'(\vec{P}) = V'_m(P_m) \dots V'_1(P_1).  
\label{eq:def-W-2}
\end{align}  
\end{definition}

We begin by justifying that it suffices to test the circuit functions on Pauli substitutions.
This step removes the continuous parameters in $\mathcal{U}$ and replaces them with a finite set $\mathcal{P}$.


\begin{lemma}
Let $\mathcal{U}$ be the set of all single-qubit unitary operators.
There exists a constant phase $\phi \in \mathbb{R}$ such that
\begin{align}
F'(\vec{U}) = e^{i\phi} F(\vec{U}), \quad \forall \, \vec{U} \in \mathcal{U}^m
\label{eq:equiv-U}
\end{align}
if and only if
\begin{equation}
F'(\vec{P})
= e^{i\phi}\, F(\vec{P}),
\quad
\forall\, \vec{P} \in \mathcal{P}^{m},
\label{eq:equiv-P}
\end{equation}
holds with the same phase $\phi$, where $\mathcal{P} = \{I,X,Y,Z\}$.
$F',F$ are defined in Definition \ref{def:F,Fprime}.
\end{lemma}

\begin{proof}
First, assume Eq.~\eqref{eq:equiv-U} holds for all $\vec{U}$. Since the set of Pauli strings $\mathcal{P}^{m}$ is a subset of all possible unitary tuples, the relation must specifically hold for any $\vec{U} = \vec{P} \in \mathcal{P}^{m}$. Thus, \eqref{eq:equiv-P} holds trivially.

Next, assume \eqref{eq:equiv-P} holds.
Each $U_i^{(k)}$ can be uniquely expanded in the $n$-qubit Pauli basis
$\{I_n, X^{(k)}, Y^{(k)}, Z^{(k)}\}$ as
\begin{equation}
U_i^{(k)} = \alpha_i I_n + \beta_i X^{(k)} + \gamma_i Y^{(k)} + \delta_i Z^{(k)}.
\label{eq:pauli-expansion}
\end{equation}
Using this expansion, the circuits can be expressed as
\begin{align}
F(\vec{U}) 
&= \sum_{\vec{P} \in \mathcal{P}^{\otimes m}}
\epsilon_{\vec{P}}\, F(\vec{P}) \label{eq:circuit-expansion}\\
F'(\vec{U}) 
&= \sum_{\vec{P} \in \mathcal{P}^{\otimes m}}
\epsilon_{\vec{P}}\, F'(\vec{P}).
\label{eq:circuit-expansion-prime}    
\end{align}
Here, $\epsilon_{\vec{P}}$ is the product of the coefficients associated with the Pauli operators constituting $\vec{P} = (P_1, \dots, P_m)$. Specifically, if we define the coefficient function $c_i(P)$ for the $i$-th unitary $U_i^{(k)}$ based on Eq.~\eqref{eq:pauli-expansion} as
\begin{equation}
c_i(P) = 
\begin{cases}
\alpha_i & \text{if } P = I_n, \\
\beta_i  & \text{if } P = X^{(k)}, \\
\gamma_i & \text{if } P = Y^{(k)}, \\
\delta_i & \text{if } P = Z^{(k)},
\end{cases}
\end{equation}
then $\epsilon_{\vec{P}}$ is given by
\begin{equation}
\epsilon_{\vec{P}} = \prod_{i=1}^m c_i(P_i).
\label{eq:def-epsilon}
\end{equation}
Substituting \eqref{eq:equiv-P} into the Pauli expansion in Eq.~\eqref{eq:circuit-expansion-prime}, we obtain
\begin{align*}
F'(\vec{U})
&=
\sum_{\vec{P} \in \mathcal{P}^{\otimes m}}
\epsilon_{\vec{P}}\, F'(\vec{P}) \\
&=
\sum_{\vec{P} \in \mathcal{P}^{\otimes m}}
\epsilon_{\vec{P}}\, e^{i\phi} F(\vec{P}) \\
&=
e^{i\phi}
\sum_{\vec{P} \in \mathcal{P}^{\otimes m}}
\epsilon_{\vec{P}}\, F(\vec{P})
=
e^{i\phi} F(\vec{U}).
\end{align*}
The last equality follows from \eqref{eq:circuit-expansion}.
\end{proof}

By this lemma, we can perform equivalence-checking with Pauli-substituted instances $F(\vec{P})$ and $F'(\vec{P})$.
The next step is to rewrite the equation for easy handling.

\begin{lemma}
\label{lem:WW=CC}
Equation \eqref{eq:equiv-P} holds 
if and only if
\begin{equation}
W^\dagger(\vec{P})\, W'(\vec{P})
= e^{i\phi}\, C_{m:0} C_{m:0}'^\dagger, 
\quad
\forall\, \vec{P}\in \mathcal{P}^{m}.
\label{eq:equiv-WW=CC}
\end{equation}
\end{lemma}

\begin{proof}
By definition of the cumulative adjoint operator $W(\vec{P})$, 
\begin{align}
    F(\vec{P}) = W(\vec{P})\, C_{m:0}, \notag \\
    F'(\vec{P}) = W'(\vec{P})\, C'_{m:0}, \notag
\end{align}
holds for any $\vec{P}$.
Since $W(\vec{P})$, $W'(\vec{P})$, $C_{m:0}$, and $C'_{m:0}$ are all unitary operators, they are invertible. 
Thus, for any $\vec{P}$, the condition in Eq. \eqref{eq:equiv-P} is equivalent to Eq.~\eqref{eq:equiv-WW=CC} through the following algebraic transformations:
\begin{align*}
&F'(\vec{P}) = e^{i\phi} F(\vec{P}) \\
&\iff \quad 
W'(\vec{P})\, C'_{m:0} = e^{i\phi} W(\vec{P})\, C_{m:0} \\
&\iff \quad 
W^\dagger(\vec{P})\, W'(\vec{P})\, C'_{m:0} = e^{i\phi} C_{m:0} 
\\
&\iff \quad 
W^\dagger(\vec{P})\, W'(\vec{P}) = e^{i\phi} C_{m:0} C'^\dagger_{m:0} 
.
\end{align*}
This concludes the proof.
\end{proof}

Lemma~\ref{lem:WW=CC} rewrites the function-equivalence condition as an operator identity in which only the left-hand side depends on the Pauli substitution $\vec{P}$.
Next, we evaluate the $\vec{P}$-independent Clifford backbones by considering the trivial substitution case.

\begin{lemma}
\label{lem:CC-WW}
Equation \eqref{eq:equiv-WW=CC} holds if and only if the following equations hold simultaneously:
\begin{equation}
  C'_{m:0} = e^{i\phi} C_{m:0},
\label{eq:equiv-CC}
\end{equation}
\begin{equation}  
W(\vec{P}) = W'(\vec{P}),
\quad
\forall\, \vec{P} \in \mathcal{P}^{m}.
\label{eq:equiv-WW}
\end{equation}
\end{lemma}

\begin{proof}
If Eq.~\eqref{eq:equiv-WW=CC} holds, by considering the trivial case where all substituted operators are identities, i.e., $\vec{P}= [I_n, \ldots, I_n]$, we find
\[
e^{i\phi}\, C_{m:0} C'^\dagger_{m:0} = I_n,
\]
or equivaletly, Eq.~\eqref{eq:equiv-CC}.
This is because, from the definition in Eq.~\eqref{eq:def-V}, we have $V_i(I_n) = I_n$ and $V'_i(I_n) = I_n$ for all $i$ and consequently the cumulative products (Eq.~\eqref{eq:def-W-1}, ~\eqref{eq:def-W-2}) become $W(\vec{I}) = I_n$ and $W'(\vec{I}) = I_n$.
By substituting Eq.~\eqref{eq:equiv-CC} into Eq.~\eqref{eq:equiv-WW=CC}, the term on the right-hand side reduces to the identity operator $I_n$. This implies $W^\dagger(\vec{P}) W'(\vec{P}) = I_n$, or equivalently,
\[
W'(\vec{P}) = W(\vec{P}),
\quad
\forall\, \vec{P} \in \mathcal{P}^{m}.
\]
These arguments show that if Eq.~\eqref{eq:equiv-WW=CC} holds Eqs.~ \eqref{eq:equiv-CC} and \eqref{eq:equiv-WW} hold.
It is also trivial to see Eqs.~\eqref{eq:equiv-CC} and \eqref{eq:equiv-WW} imply Eq. \eqref{eq:equiv-WW=CC}.
We now have that Eq.~\eqref{eq:equiv-WW=CC} holds if and only if Eqs.~ \eqref{eq:equiv-CC} and \eqref{eq:equiv-WW} hold.
\end{proof}

At this point, the verification problem splits cleanly into two parts:
(i) equality of the Clifford backbones up to a global phase, Eq.~\eqref{eq:equiv-CC}, and
(ii) equality of the cumulative adjoint products for all Pauli substitutions, Eq.~\eqref{eq:equiv-WW}.
We now show that the latter condition is equivalent to a simple layerwise equality of the adjoint images $V_i$ and $V'_i$.

\begin{lemma}
\label{lem:VV}
Equation \eqref{eq:equiv-WW} holds if and only if
\begin{equation}  
  V_i(P_i) = V'_i(P_i),
  \qquad \forall \, i \in \{1,\dots,m\},
  \quad \forall \, P_i \in \mathcal{P}.
\label{eq:equiv-VV}
\end{equation}
\end{lemma}

\begin{proof}
If Eq~\eqref{eq:equiv-VV} holds for all $i$, then the ordered products defining $W(\vec{P})$ and $W'(\vec{P})$ clearly coincide for any $\vec{P}$, and hence $W(\vec{P}) = W'(\vec{P})$.

In contrast, we show that Eq~\eqref{eq:equiv-WW} forces the equality $V_i(P_i) = V'_i(P_i)$ to hold at every layer.
Indeed, if there existed an index $i$ and a Pauli operator $P$
such that $V_i(P_i) \neq V'_i(P_i)$ then you could isolate this mismatch by choosing a Pauli-substitution sequence of the form
\[
\vec{P} = [I_n, \ldots, I_n, P_i, I_n, \ldots, I_n],
\]
with a single nontrivial entry at position $i$.
For this choice, all other $V(P)$ and $V'(P)$ in the products $W(\vec{P})$ and $W'(\vec{P})$ reduce to the identity $I_n$, and hence
\[
W(\vec{P}) = V_i(P_i), \qquad
W'(\vec{P}) = V'_i(P_i).
\]
The assumed equality $W(\vec{P}) = W'(\vec{P})$ for any arrangement $\vec{P}$ would then immediately imply $V_i(P_i) = V'_i(P_i)$, contradicting the premise.
Therefore, if Eq~\eqref{eq:equiv-WW} holds, we conclude that $V_i(P_i) = V'_i(P_i)$ holds for all layers $i$.
This establishes the converse implication and completes the proof that Eqs.~\eqref{eq:equiv-WW} and \eqref{eq:equiv-VV} are equivalent.
\end{proof}

It remains to express the layerwise operator equalities in a form that can be checked efficiently.
Since each $V_i(P),V'_i(P_i)$ is the $n$-qubit Pauli operator, it admits a compact binary symplectic representation.
By collecting these representations for $P=X$ and $P=Z$ across all layers, we obtain the decision matrix $B$ (and analogously $B'$), so that the equality of all adjoint images reduces to the matrix equality $B=B'$.

\begin{lemma}
Equation \eqref{eq:equiv-VV} holds if and only if
\begin{align}
B' = B.
\label{eq:equiv-BB}
\end{align}
$B,B'$ are defined in Definition \eqref{def:B,Bprime}.
\end{lemma}
\begin{proof}
Assume that Eq.~\eqref{eq:equiv-VV} holds for $P=Z$ and $P=X$, so that $V_i(Z)=V'_i(Z)$ and $V_i(X)=V'_i(X)$ for every $i$. 
Under definition \ref{def:V}, this yields $\bm v_i(Z)=\bm v'_i(Z)$ and $\bm v_i(X)=\bm v'_i(X)$ for all $i$.
Under definition \ref{def:B,Bprime}, every column of $B$ therefore equals the same column of $B'$. 
Hence, Eq.~\eqref{eq:equiv-BB} holds.

Conversely, we assume Eq.~\eqref{eq:equiv-BB}.
Under definition \ref{def:V} and \ref{def:B,Bprime}, the equality $B'=B$ implies that
\[
V_i(Z)=V'_i(Z)\quad\text{and}\quad V_i(X)=V'_i(X)\qquad \forall i.
\]
It then follows that the $Y$-case also coincides:
using $Y = iXZ$,
\begin{align}
V_i(Y)&=C_{m:i}Y^{(k_i)}C_{m:i}^\dagger \notag \\
&=i\,C_{m:i}X^{(k_i)}C_{m:i}^\dagger\ C_{m:i}Z^{(k_i)}C_{m:i}^\dagger \notag \\
&=i\,V_i(X)V_i(Z), \notag
\end{align}
and similarly $V'_i(Y)=i\,V'_i(X)V'_i(Z)$.
Thus $V_i(X)=V'_i(X)$ and $V_i(Z)=V'_i(Z)$ imply $V_i(Y)=V'_i(Y)$.
Finally, $P=I$ is trivial since $V_i(I)=V'_i(I)=I_n$ by definition.
Therefore, Eq.~\eqref{eq:equiv-VV} holds.
Hence, Eq.~\eqref{eq:equiv-VV} holds if and only if Eq.~\eqref{eq:equiv-BB}.
\end{proof}

Combining the above lemmas, we obtain a complete reduction from template equivalence
$F'(\vec{U})=e^{i\phi}F(\vec{U})$ for all $\vec{U}\in\mathcal{U}^m$
to the simultaneous conditions $C'_{m:0}=e^{i\phi}C_{m:0}$ and $B'=B$.
This completes the proof of the main theorem.


\section{Rearrangement Case}
\label{app:rearrange}
In this section, we extend the previous analysis to the case where the single-qubit layers in $F'$ appear in a permuted order.
We fix a permutation $\sigma\in S_m$ describing this rearrangement and introduce the notation that tracks (i) which layer in $F$ corresponds to which layer in $F'$, and (ii) which pairs of layers must be swapped when realizing the permutation $\sigma$ as a product of adjacent transpositions.
These notions will be used to state commuting conditions that guarantee that reordering does not introduce an additional sign in the cumulative Pauli products.


\begin{definition}[Correspondence Set]
\label{def:correspondence-set}
To clarify the layer mapping induced by $\sigma$, we define the correspondence set $\mathcal{M}$ as
\begin{align}
\mathcal{M} := \big\{ (i, j) \in \{1, \dots, m\}^2 \mid \sigma(j) = i \big\}, \notag
\end{align}
which identifies that the $i$-th unitary in $F$ corresponds to the $j$-th unitary in $F'$.
\end{definition}

The set $\mathcal{M}$ allows us to compare the same single-qubit gate $U_i$ (or Pauli substitution $P_i$) across the two circuits even when their order differ.
To control the effect of reordering on products of Pauli operators, we also require a record of the elementary adjacent swaps needed to implement $\sigma$.

\begin{definition}[Set of Adjacent Transpositions]
\label{def:trans}
Let $\sigma \in S_m$ be a permutation.
We define $T(\sigma)$ as the set of adjacent transpositions $(r,r')$ that appear in a decomposition of $\sigma$ of the form
\begin{align}
\sigma = \tau_1 \tau_2 \cdots \tau_k,
\qquad \tau_j = (r_j,r'_j),
\end{align}
where $k$ is the minimum possible number of adjacent transpositions required to realize $\sigma$.
That is,
\begin{align}
T(\sigma) := \{\, (r_j,r'_j) \mid j = 1,\ldots,k \,\}.
\end{align}
\end{definition}

The set $T(\sigma)$ specifies the (minimal) adjacent swaps that occur when $\sigma$ is realized by adjacent transpositions.
Since we evaluate commutation relations between propagated Pauli operators, it is convenient to understand commutation in the binary symplectic representation appearing in the Decision Matrix.

\begin{definition}[Symplectic Inner Product]
For two Pauli operators $P,P'$ with binary symplectic representations defined as Definition \ref{def:V},
\[
\bm{v}(P) = (z,x,r), \qquad \bm{v}'(P') = (z',x',r'),
\]
their \emph{symplectic inner product} is defined as
\begin{align}
\langle \bm{v}(P), \bm{v}'(P') \rangle
:&= z \cdot x' \oplus x \cdot z' \notag \\
&= \sum_{q=1}^n (z_q x'_q \oplus x_q z'_q) \in \mathbb{F}_2.  
\end{align}

\end{definition}

Two Pauli operators $P$ and $P'$ commute if and only if
\[
\langle \bm{v}(P), \bm{v}'(P') \rangle = 0.
\]

We will use this inner product to check, in a purely binary manner, whether two propagated Pauli operators $P_r, P_{r'}$ commute.
In particular, for each adjacent swap encoded by $T(\sigma)$, we will impose commutation of the corresponding propagated operators so that the reordering produces no extra sign.
We now formalize the permuted circuit functions.

\begin{definition}[Problem Setting at Rearrangement Case]
\label{def:F,Fprime-perm}
Let $\vec{U} = (U_1, \dots, U_m)$ be a sequence of single-qubit unitaries.
We consider a permutation $\sigma \in S_m$ on the index set $\{1, \dots, m\}$ and define the permuted sequence $\vec{U}_\sigma$ as
\begin{equation}
\vec{U}_\sigma := \sigma (\vec{U}) = \big( U_{\sigma(1)}, \dots, U_{\sigma(m)} \big), 
\end{equation}
where $U_{\sigma(i)}$ denotes the $i$-th unitary in the permuted sequence.
The quantum circuits $F$ and $F'$ are defined as
\begin{align}
F(\vec{U}) &:= C_m U_m^{(k_m)} C_{m-1} U_{m-1}^{(k_{m-1})} \cdots C_1 U_1^{(k_1)} C_0,
\label{eq:def-circuit_3} \\
F'(\vec{U_\sigma}) &:= C'_m U_{\sigma(m)}^{(l_m)} C'_{m-1} U_{\sigma(m-1)}^{(l_{m-1})} \cdots C'_1 U_{\sigma(1)}^{(l_1)} C'_0.
\label{eq:def-circuit_4}
\end{align}
\end{definition}

As in the non-permuted case, we analyze equivalence-as-function by substituting Pauli operators for the single-qubit unitaries and propagating them through the Clifford layers.
The only difference is that the same Pauli substitution $P_i$ at layer $i$ of $F$ must be compared with the corresponding layer $j$ of $F'$ determined by $\mathcal{M}$.

\begin{definition}[Adjoint Images at Rearrangement Case]
\label{def:V-perm}
For each corresponding pair $(i,j) \in \mathcal{M}$, we define the adjoint stabilizer
operators associated with substituting the same Pauli operator $P_i$ into the
$i$-th layer of $F$ and the $j$-th layer of $F'$, respectively, as
\begin{align}
V_i(P_i) := C_{m:i}\, P_i^{(k_i)}\, C_{m:i}^\dagger,
\\
V'_j(P_i) := C_{m:j}\, P_i^{(l_j)}\, C'^\dagger_{m:j}.
\label{eq:def-V-perm}
\end{align}
And, we also define the cumulative adjoint products
\begin{align}
W(\vec{P}) = V_m(P_m) \dots V_1(P_1),
\label{eq:def-W-perm-1}
\\
W'(\vec{P}_\sigma) = V'_m(P_{\sigma(m)}) \dots V'_1(P_{\sigma(1)}).  
\label{eq:def-W-perm-2}
\end{align}   
\end{definition}

The operators $V_i(P_i)$ and $V'_j(P_i)$ capture how the same single Pauli substitution $P_i$ is conjugated by the Clifford operators in each circuit.
Their ordered products $W(\vec{P})$ and $W'(\vec{P}_\sigma)$ collect the cumulative effects of substitutions, and the permutation $\sigma$ manifests itself only as a reordering of these factors.
To efficiently state equivalence-checking conditions, we next rewrite the decision matrix with the column permutation, induced by $\sigma$.

\begin{definition}[Decision Matrix at Rearrangement Case]
\label{def:B-perm,B-prime}
For the Decision Matrix $B$ defined as Definition \ref{def:B,Bprime}, we define the permutated Decision Matrix $B_\sigma$ of the circuit $F$ defined as Definition \ref{def:F,Fprime-perm},
\begin{align}
B_\sigma &:= 
\begin{bmatrix}
B_{:,\sigma(1)} & B_{:,\sigma(2)} & \cdots & B_{:,\sigma(m)}
\label{eq:def-B-perm}
\end{bmatrix},
\end{align}
where $B_{:,i}$ is $i$-th column of $B$.
\end{definition}

We now prove the completeness of our method in permuted case.
As in the non-permuted case, we first replace arbitrary single-qubit unitaries by Pauli substitutions.
We then isolate the Clifford backbones and collect all substitution-dependent terms into cumulative adjoint products.
Finally, we consider the effect of the permutation $\sigma$ on the ordering of Pauli factors by evaluating commutation of a minimal adjacent-transposition decomposition of $\sigma$.

\begin{lemma}
There exists a constant phase $\phi \in \mathbb{R}$ such that
\begin{align}
F'(\vec{U}_\sigma) = e^{i\phi} F(\vec{U}), 
\quad \forall\, \vec{U}\in\mathcal{U}^m, \, \vec{U}_\sigma := \sigma(\vec{U}),
\label{eq:equiv-U-perm}
\end{align}
if and only if
\begin{equation}
F'(\vec{P}_\sigma) = e^{i\phi}\, F(\vec{P}),
\quad \forall\, \vec{P} \in \mathcal{P}^{m},\, \vec{P}_\sigma := \sigma(\vec{P}),
\label{eq:equiv-P-perm}
\end{equation}
holds with the same phase $\phi$, where 
\begin{align}
\vec{P}_\sigma := \sigma (\vec{P}) = \big( P_{\sigma(1)}, \dots, P_{\sigma(m)} \big),
\label{eq:P-vec}
\end{align}
\end{lemma}

\begin{proof}
First, assume Eq.~\eqref{eq:equiv-U-perm} holds for all $\vec{U}$. Since the set of Pauli strings $\mathcal{P}^{m}$ is a subset of all possible unitary tuples, the relation must specifically hold for any $\vec{U} = \vec{P} \in \mathcal{P}^{m}$. Thus, \eqref{eq:equiv-P-perm} holds trivially.

Next, assume \eqref{eq:equiv-P-perm} holds.
Each $U_i^{(k)}$ can be uniquely expanded in the $n$-qubit Pauli basis
$\{I_n, X^{(k)}, Y^{(k)}, Z^{(k)}\}$ as
\begin{equation}
U_i^{(k)} = \alpha_i I_n + \beta_i X^{(k)} + \gamma_i Y^{(k)} + \delta_i Z^{(k)}.
\label{eq:pauli-expansion-perm}
\end{equation}
Using this expansion, the circuits can be expressed as
\begin{align}
F(\vec{U}) 
&= \sum_{\vec{P} \in \mathcal{P}^{m}}
\epsilon_{\vec{P}}\, F(\vec{P}), \label{eq:circuit-expansion-perm}\\
F'(\vec{U}_\sigma) 
&= \sum_{\vec{P}_\sigma \in \mathcal{P}^{m}}
\epsilon_{\vec{P}}\, F'(\vec{P}_\sigma).
\label{eq:circuit-expansion-prime-perm}    
\end{align}
Here, $\epsilon_{\vec{P}}$ is the product of the coefficients associated with the Pauli operators constituting $\vec{P} = (P_1, \dots, P_m)$. Specifically, if we define the coefficient function $c_i(P)$ for the $i$-th unitary $U_i^{(k)}$ based on Eq.~\eqref{eq:pauli-expansion-perm} as
\begin{equation}
c_i(P) = 
\begin{cases}
\alpha_i & \text{if } P = I_n, \\
\beta_i  & \text{if } P = X^{(k)}, \\
\gamma_i & \text{if } P = Y^{(k)}, \\
\delta_i & \text{if } P = Z^{(k)},
\end{cases}
\end{equation}
then $\epsilon_{\vec{P}}$ is given by
\begin{equation}
\epsilon_{\vec{P}} = \prod_{i=1}^m c_i(P_i).
\label{eq:def-epsilon-perm}
\end{equation}
Substituting \eqref{eq:equiv-P-perm} into the Pauli expansion in Eq.~\eqref{eq:circuit-expansion-prime-perm}, we obtain
\begin{align*}
F'(\vec{U}_\sigma)
&=
\sum_{\vec{P}_\sigma \in \mathcal{P}^{\otimes m}}
\epsilon_{\vec{P}}\, F'(\vec{P}_\sigma) \\
&=
\sum_{\vec{P} \in \mathcal{P}^{\otimes m}}
\epsilon_{\vec{P}}\, e^{i\phi} F(\vec{P}) \\
&=
e^{i\phi}
\sum_{\vec{P} \in \mathcal{P}^{\otimes m}}
\epsilon_{\vec{P}}\, F(\vec{P})
=
e^{i\phi} F(\vec{U}).
\end{align*}
The last equality follows from \eqref{eq:circuit-expansion-perm}.
\end{proof}

By this lemma, it suffices to analyze the Pauli-substituted instances $F(\vec{P})$ and $F'(\vec{P}_\sigma)$.
The next step is to separate the $\vec{P}$-dependence from the Clifford backbones by rewriting each circuit as a product of a cumulative adjoint operator and the total Clifford part.

\begin{lemma}
\label{lem:WW=CC-perm}
Equation \eqref{eq:equiv-P-perm} holds 
if and only if
\begin{equation}
W^\dagger(\vec{P})\, W'(\vec{P}_\sigma)
= e^{i\phi}\, C_{m:0} C_{m:0}'^\dagger, 
\quad
\forall\, \vec{P}\in \mathcal{P}^{m}, \,  \vec{P}_\sigma := \sigma(\vec{P}).
\label{eq:equiv-WW=CC-perm}
\end{equation}
\end{lemma}

\begin{proof}
By definition of the cumulative adjoint operator $W(\vec{P})$, 
\begin{align}
F(\vec{P}) = W(\vec{P})\, C_{m:0}, \notag \\
F'(\vec{P}_\sigma) = W'(\vec{P}_\sigma)\, C'_{m:0}, \notag
\label{eq:equiv-WC-perm}
\end{align}
holds for any $\vec{P}$.
Since $W(\vec{P})$, $W'(\vec{P})$, $C_{m:0}$, and $C'_{m:0}$ are all unitary operators, they are invertible. 
Thus, for any $\vec{P}$, the condition in Eq. \eqref{eq:equiv-P-perm} is equivalent to Eq.~\eqref{eq:equiv-WW=CC-perm} through the following algebraic transformations:
\begin{align*}
&F'(\vec{P}_\sigma) = e^{i\phi} F(\vec{P}) \\
&\iff \quad 
W'(\vec{P}_\sigma)\, C'_{m:0} = e^{i\phi} W(\vec{P})\, C_{m:0} \\
&\iff \quad 
W^\dagger(\vec{P})\, W'(\vec{P}_\sigma)\, C'_{m:0} = e^{i\phi} C_{m:0} 
\\
&\iff \quad 
W^\dagger(\vec{P})\, W'(\vec{P}_\sigma) = e^{i\phi} C_{m:0} C'^\dagger_{m:0} 
.
\end{align*}
This concludes the proof.
\end{proof}

Lemma~\ref{lem:WW=CC-perm} expresses equivalence under Pauli substitutions as an identity of the form
$W^\dagger(\vec{P})W'(\vec{P}_\sigma)=e^{i\phi}C_{m:0}C_{m:0}'^\dagger$.
Next, we evaluate the $P$-independent Clifford backbones by considering the
trivial substitution case as in non-permuted case.

\begin{lemma}
\label{lem:CC-WW-perm}
Equation \eqref{eq:equiv-WW=CC-perm} holds if and only if the following equations hold simultaneously:
\begin{equation}
  C' = e^{i\phi} C_{m:0},
\label{eq:equiv-CC-perm}
\end{equation}
\begin{equation}  
W(\vec{P}) = W'(\vec{P}_\sigma),
\quad
\forall\, \vec{P} \in \mathcal{P}^{m}, \,  \vec{P}_\sigma := \sigma(\vec{P}).
\label{eq:equiv-WW-perm}
\end{equation}
\end{lemma}

\begin{proof}
If Eq.~\eqref{eq:equiv-WW=CC-perm} holds, by considering the trivial case where all substituted operators are identities, i.e., $\vec{P} = \vec{P}_\sigma = [I_n, \ldots, I_n]$, we find
\[
e^{i\phi}\, C_{m:0} C'^\dagger_{m:0} = I_n,
\]
or, equivaletly, Eq.~\eqref{eq:equiv-CC-perm}.
This is because, from the definition in Eq.~\eqref{eq:def-V-perm}, we have $V_i(I_n) = I_n$ and $V'_i(I_n) = I_n$ for all $i$ and consequently the cumulative products (Eq.~\eqref{eq:def-W-perm-1}, ~\eqref{eq:def-W-perm-2}) become $W(\vec{I}) = I_n$ and $W'(\vec{I}) = I_n$.
By substituting Eq.~\eqref{eq:equiv-CC-perm} into Eq.~\eqref{eq:equiv-WW=CC-perm}, the term on the right-hand side reduces to the identity operator $I_n$. This implies $W^\dagger(\vec{P}) W'(\vec{P}) = I_n$, or equivalently,
\[
W(\vec{P}) = W'(\vec{P}_\sigma),
\quad
\forall\, \vec{P} \in \mathcal{P}^{m}, \, \vec{P}_\sigma := \sigma(\vec{P}).
\]
These arguments show that if Eq.~\eqref{eq:equiv-WW=CC-perm} holds Eqs.~ \eqref{eq:equiv-CC-perm} and \eqref{eq:equiv-WW-perm} hold.
It is also trivial to see Eqs.~\eqref{eq:equiv-CC-perm} and \eqref{eq:equiv-WW-perm} imply Eq. \eqref{eq:equiv-WW=CC-perm}.
We now have that Eq.~\eqref{eq:equiv-WW=CC-perm} holds if and only if Eqs.~ \eqref{eq:equiv-CC-perm} and \eqref{eq:equiv-WW-perm} hold.
\end{proof}

At this point the permuted-template equivalence problem again splits into two parts:
(i) equality of the Clifford backbones up to a global phase, Eq.~\eqref{eq:equiv-CC-perm}, and
(ii) equality of the cumulative adjoint products under permutation, Eq.~\eqref{eq:equiv-WW-perm}.
The remaining difficulty is that $W(\vec{P})$ and $W'(\vec{P}_\sigma)$ are products of Pauli operators whose factors are the same but whose order is permuted.
We next show that equality of these ordered products for all $\vec{P}$ is equivalent to a layerwise matching condition together with commutation constraints in the adjacent swaps in $T(\sigma)$.

\begin{lemma}
\label{lem:VV-perm}
Equation \eqref{eq:equiv-WW-perm} and holds if and only if the following equations hold simultaneously:
\begin{align}  
V_i(P_i) &= V'_j(P_i),
\qquad \forall \, (i,j) \in  \mathcal{M}, \quad \forall \, P_i \in \mathcal{P}, 
\label{eq:equiv-VV-perm} \\
V_r(P_r)V_{r'}(P_{r'}) &= V_{r'}(P_{r'})V_r(P_r), 
\qquad \forall \, (r,r') \in  T(\sigma).
\label{eq:equiv-Vtrans-perm}
\end{align}
\end{lemma}

\begin{proof}
If Eq.~\eqref{eq:equiv-VV-perm} holds for every corresponding pair $(i,j)\in\mathcal{M}$ and every $P\in\mathcal{P}$, the Pauli (stabilizer) factors defined in Definition~\ref{def:V-perm} coincide.
Consequently, $W(\vec{P})$ and $W'(\vec{P}_\sigma)$ are products of the same set of Pauli (stabilizer) factors, differing only by their order.
However, Pauli operators either commute or anticommute, so equality of the individual factors alone does not in general imply equality of the ordered products.
Now assume in addition Eq.~\eqref{eq:equiv-Vtrans-perm}.
Each transposition $\tau_j$ in the $\sigma$, defined in Definition \ref{def:trans} swaps two neighboring factors, and Eq.~\eqref{eq:equiv-Vtrans-perm} guarantees that every such swapped pair commutes.
Therefore, the product is invariant under all swaps required to reorder $W(\vec{P})$ by $\sigma$, and hence Eq.~\eqref{eq:equiv-WW-perm} holds.
Thus, Eqs.~\eqref{eq:equiv-VV-perm} and \eqref{eq:equiv-Vtrans-perm} together imply Eq.~\eqref{eq:equiv-WW-perm}.

In contrast, we assume Eq.~\eqref{eq:equiv-WW-perm}.
Consider the choice
\begin{align}
\vec{P} = [I_n, \ldots, I_n, \underbrace{P_i}_{\text{$i$-th}}, I_n, \ldots, I_n], \notag
\\
\vec{P}_\sigma = [I_n, \ldots, I_n, \underbrace{P_i}_{\text{$j$-th}}, I_n, \ldots, I_n]. \notag
\end{align}

By Definition~\ref{def:V-perm}, all factors in $W(\vec{P})$ (resp., $W'(\vec{P}_\sigma)$) are the identity except the single nontrivial factor $V_i(P_i)$ (resp., $V'_j(P_i)$), hence
$W(\vec{P})=V_i(P_i)$ and $W'(\vec{P}_\sigma)=V'_j(P_i)$.
Therefore Eq.~\eqref{eq:equiv-WW-perm} implies $V_i(P_i)=V'_j(P_i)$, and since $(i,j)$ and $P_i$ were arbitrary, Eq.~\eqref{eq:equiv-VV-perm} follows.
Next, we assume for contradiction that there exist $P_r,P_{r'}\in\mathcal{P}$ such that $V_a(P_a)$ and $V_{b}(P_{b})$ anticommute, i.e., $V_a(P_a)V_{b}(P_{b})=-\,V_{b}(P_{b})V_a(P_a)$.
Consider the choice
\begin{align}
\vec{P} = [I_n, \ldots, I_n, \underbrace{P_a}_{\text{$a$-th}}, I_n, \ldots, I_n, \underbrace{P_b}_{\text{$b$-th}}, I_n, \ldots, I_n], \notag
\\
\vec{P}_\sigma = [I_n, \ldots, I_n, \underbrace{P_b}_{\text{$a$-th}}, I_n, \ldots, I_n, \underbrace{P_a}_{\text{$b$-th}}, I_n, \ldots, I_n]. \notag
\end{align}

Then $W(\vec{P})$ reduces to $V_a(P_a)V_b(P_b)$.
On the other hand, since $(a,b)\in T(\sigma)$ appears as a transposition in the decomposition realizing $\sigma$, the corresponding $\sigma$-ordered product swaps the relative order of these two nontrivial factors, and thus $W'(\vec{P}_\sigma)$ reduces to $V_b(P_b)V_a(P_a)$.
Using Eq.~\eqref{eq:equiv-VV-perm} to identify the corresponding factors in $W'(\vec{P}_\sigma)$ with those in $W(\vec{P})$, Eq.~\eqref{eq:equiv-WW-perm} would force 
\[
W(\vec{P}) = V_a(P_a)V_b(P_b)=V_b(P_b)V_a(P_a) = W'(\vec{P}_\sigma),
\]
contradicting anticommutation.
Therefore $V_a(P_a)$ and $V_b(P_b)$ commute for all $P_a,P_b\in\mathcal{P}$ and all $(a,b)\in T(\sigma)$, which is Eq.~\eqref{eq:equiv-Vtrans-perm}.
\end{proof}

Lemma~\ref{lem:VV-perm} reduces the permuted cumulative condition $W(\vec{P})=W'(\vec{P}_\sigma)$ to two purely local requirements:
a correspondence condition $V_i(P)=V'_j(P)$ for each $(i,j)\in\mathcal{M}$, and commutation conditions for the swapped pairs in $T(\sigma)$.
To obtain an efficiently checkable criterion, we now express the correspondence condition in the binary symplectic representation and show that it is equivalent to equivalence of the decision matrices up to the column permutation induced by $\sigma$.

\begin{lemma}
Equation \eqref{eq:equiv-VV-perm} holds if and only if
\begin{align}
B' = B_\sigma.
\label{eq:equiv-BB-perm}
\end{align}
$B'$ is defined in Definition \ref{def:B,Bprime}, $B_\sigma$ is defined in Definition \ref{def:B-perm,B-prime}.
\end{lemma}

\begin{proof}
Assume that Eq.~\eqref{eq:equiv-VV-perm} holds for $P=Z$ and $P=X$, so that $V_i(Z)=V'_j(Z)$ and $V_i(X)=V'_j(X)$ for every layer. 
Under Definition \ref{def:V}, this yields $\bm v_i(Z)=\bm v'_j(Z)$ and $\bm v_i(X)=\bm v'_j(X)$.
Under Definition \ref{def:B,Bprime}, $i$-th column of $B$ therefore equals the $j$-th column of $B'$.
Hence, under Definition \ref{def:B-perm,B-prime} and $\sigma(j) = i$ from Definition \ref{def:correspondence-set}, $i$-th column of $B$ equals the $j$-th column of $B_\sigma$.
Thus, Eq.~\eqref{eq:equiv-BB-perm} holds.

Conversely, we assume Eq.~\eqref{eq:equiv-BB-perm}.
Under Definition \ref{def:V}, \ref{def:B,Bprime}, and \ref{def:B-perm,B-prime}, the equality $B'=B_\sigma$ implies that
\[
V_i(Z)=V'_j(Z)\quad\text{and}\quad V_i(X)=V'_j(X)\qquad \forall \, (i,j) \in \mathcal{M}.
\]
It then follows that the $Y$-case also coincides:
using $Y = iXZ$,
\begin{align}
V_i(Y)&=C_{m:i}Y^{(k_i)}C_{m:i}^\dagger \notag \\
&=i\,C_{m:i}X^{(k_i)}C_{m:i}^\dagger\ C_{m:i}Z^{(k_i)}C_{m:i}^\dagger \notag \\
&=i\,V_i(X)V_i(Z), \notag
\end{align}
and similarly $V'_j(Y)=i\,V'_j(X)V'_j(Z)$.
Thus $V_i(X)=V'_j(X)$ and $V_i(Z)=V'_j(Z)$ imply $V_i(Y)=V'_j(Y)$.
Finally, $P=I$ is trivial since $V_i(I)=V'_j(I)=I_n$ by definition.
Therefore, Eq.~\eqref{eq:equiv-VV-perm} holds.
Hence, Eq.~\eqref{eq:equiv-VV-perm} holds if and only if Eq.~\eqref{eq:equiv-BB-perm}.
\end{proof}

We now turn to the correspondence condition Eq.~\eqref{eq:equiv-VV-perm} and express it in terms of the Decision Matrices.
We note that the commutation constraints in Eq.~\eqref{eq:equiv-Vtrans-perm} can also be checked purely in the binary symplectic representation in Decision Mtrix via the symplectic inner product introduced above.

\begin{lemma}
\label{lem:sinplec}
Equation \eqref{eq:equiv-Vtrans-perm} holds if and only if
\begin{align}  
\langle \bm{v}_r(P_r), \bm{v}_{r'}(P_{r'}) \rangle = 0,
\qquad \forall \,(r,r')\in T(\sigma).
\label{eq:simplec}
\end{align}
\end{lemma}

\begin{proof}
For any pair of $n$-qubit Pauli operators $P,P'$, those commute if and only if 
\[
\langle \bm{v}(P),\bm{v}'(P')\rangle=0.
\]
Applying this with $V_r(P_r)$ and $V_{r'}(P_{r'})$ yields the claim for each $(r,r')\in T(\sigma)$.
\end{proof}

Therefore, the verification of Eq.~\eqref{eq:equiv-Vtrans-perm} reduces to checking vanishing symplectic inner products of the corresponding propagated Pauli operators.

Combining the above lemmas, the permuted equivalence-as-function
$F'(\vec{U}_\sigma)=e^{i\phi}F(\vec{U})$ for all $\vec{U}\in\mathcal{U}^m$
reduces to (i) the backbone condition $C'_{m:0}=e^{i\phi}C_{m:0}$,
(ii) the decision-matrix relation $B'=B_\sigma$, and 
(iii) the commutation constraints in Eq.~\eqref{eq:simplec}.
This finishes the proof of the completeness of our method for the rearrangement cases which is in more general.


\section{Heuristic Verification for Fixed Instances}
\label{app:heuristic}

So far we have established a complete (necessary-and-sufficient) criterion for \emph{function} equivalence, i.e., equality for all assignments $\vec{U}$.
For practical compiler testing, it is also useful to state a sufficient (non-necessary) criterion for \emph{fixed} instances: 
if all the Pauli-substituted identities hold uniformly with a common phase, then the instantiated circuits agree for any particular choice of $\vec{U}$.

\begin{lemma}
\label{lem:rough-sufficient}
Suppose there exists a constant phase $\phi\in\mathbb{R}$ such that
\begin{equation}
F'(\vec{P}) = e^{i\phi}\,F(\vec{P}),
\qquad \forall\,\vec{P}\in\mathcal{P}^m,
\label{eq:rough-P}
\end{equation}
holds with the same $\phi$, where $\mathcal{P}=\{I,X,Y,Z\}$.
Then, for any fixed assignment $\vec{U}\in\mathcal{U}^m$, we have
\begin{equation}
F'(\vec{U}) = e^{i\phi}\,F(\vec{U}).
\label{eq:rough-U}
\end{equation}
In particular, Eq.~\eqref{eq:rough-U} holds for some $\vec{U}\in\mathcal{U}^m$.
\end{lemma}

\begin{proof}
Fix an $\vec{U}\in\mathcal{U}^m$ to the paticular one set of sequence.
Each single-qubit unitary $U_i^{(k_i)}$ admits a (unique) expansion in the Pauli basis $\{I_n,X^{(k_i)},Y^{(k_i)},Z^{(k_i)}\}$, so by multilinearity of circuit composition we can write
\[
F(\vec{U})=\sum_{\vec{P}\in\mathcal{P}^m}\epsilon_{\vec{P}}\,F(\vec{P}),
\qquad
F'(\vec{U})=\sum_{\vec{P}\in\mathcal{P}^m}\epsilon_{\vec{P}}\,F'(\vec{P}),
\]
with the same coefficients $\epsilon_{\vec{P}}$ for both circuits.
Substituting Eq.~\eqref{eq:rough-P} into the second expansion yields
\[
F'(\vec{U})
=\sum_{\vec{P}}\epsilon_{\vec{P}}\,e^{i\phi}F(\vec{P})
=e^{i\phi}\sum_{\vec{P}}\epsilon_{\vec{P}}\,F(\vec{P})
=e^{i\phi}F(\vec{U}),
\]
which proves Eq.~\eqref{eq:rough-U}.
\end{proof}

Lemma~\ref{lem:rough-sufficient} can be viewed as a soundness statement: a uniform Pauli-substitution identity with a common phase implies equivalence of the instantiated circuits.
Unlike the equivalence-as-function lemmas above, it is one-sided and is stated separately.

\bibliography{reference.bib}

\end{document}